%% file: main.tex
\documentclass[a4paper,USenglish]{article}
\usepackage[utf8]{inputenc}
\usepackage{geometry}
\usepackage{graphicx}
\usepackage{color}
\usepackage{amsmath,amssymb,amsfonts,amsthm}
\usepackage[ruled,linesnumbered]{algorithm2e}
\usepackage{hyperref}
\usepackage[capitalise, noabbrev]{cleveref}
	\crefname{algocf}{Algorithm}{Algorithms}
	\Crefname{algocf}{Algorithm}{Algorithms}

\usepackage{enumerate}

\newtheorem{theorem}{Theorem}
	\newtheorem{lemma}[theorem]{Lemma}
	
	\newtheorem{proposition}[theorem]{Proposition}
	
	\newtheorem{definition}[theorem]{Definition}

	\theoremstyle{definition}
	
	\theoremstyle{remark}
	
	\newtheorem*{note*}{Note}
	
	\newtheorem*{remark*}{Remark}
	\theoremstyle{plain}

\newcommand{\cO}{\ensuremath{\mathcal{O}}\xspace}

\newcommand{\cT}{\ensuremath{\mathcal{T}}\xspace}
\newcommand{\cS}{\ensuremath{\mathbb{C}}\xspace}
\newcommand{\gor}{\ensuremath{\text{OR}^+}\xspace}

\title{A general framework for enumerating equivalence classes of solutions\footnote{This is a preprint version of the following published paper: Wang, Y., Mary, A., Sagot, MF. et al. A General Framework for Enumerating Equivalence Classes of Solutions. Algorithmica (2023). https://doi.org/10.1007/s00453-023-01131-1}} 

\author{Yishu Wang\footnote{
Université de Lyon, Université Lyon 1, CNRS, Laboratoire de Biométrie et Biologie Evolutive UMR 5558, F-69622 Villeurbanne, France. Inria Grenoble Rhône-Alpes, Villeurbanne, France. \href{mailto:yishu.wang@univ-lyon1.fr}{yishu.wang@univ-lyon1.fr}. \href{mailto:arnaud.mary@univ-lyon1.fr}{arnaud.mary@univ-lyon1.fr}}, 
Arnaud Mary\textsuperscript{*}, 
Marie-France Sagot\footnote{Inria Grenoble Rhône-Alpes, Villeurbanne, France. Université de Lyon, Université Lyon 1, CNRS, Laboratoire de Biométrie et Biologie Evolutive UMR 5558, F-69622 Villeurbanne, France. \href{mailto:marie-france.sagot@inria.fr}{marie-france.sagot@inria.fr}},
Blerina Sinaimeri\footnote{Luiss University, Rome, Italy. Inria Grenoble Rhône-Alpes, Villeurbanne, France. Université de Lyon, Université Lyon 1, CNRS, Laboratoire de Biométrie et Biologie Evolutive UMR 5558, F-69622 Villeurbanne, France. \href{mailto:bsinaimeri@luiss.it}{bsinaimeri@luiss.it}}
}
\date{}

\begin{document}

\maketitle

\begin{abstract}
When a problem has more than one solution, it is often important, depending on the underlying context, to enumerate (i.e., to list) them all. Even when the enumeration can be done in polynomial delay, that is, spending no more than polynomial time to go from one solution to the next, this can be costly as the number of solutions themselves may be huge, including sometimes exponential. Furthermore, depending on the application, many of these solutions can be considered equivalent. The problem of an efficient enumeration of 
the equivalence classes or of one representative per class (without generating all the solutions), although identified as a need in many areas, has been addressed only for very few specific cases.  In this paper, we provide a general framework that solves this  problem in polynomial delay for a wide variety of  
contexts, including optimization ones that can be addressed by dynamic programming algorithms, and for certain types of equivalence relations between solutions.   
\end{abstract}

\section{Introduction}\label{section:introduction}

Enumerating the solutions of an optimization problem solved by a dynamic programming algorithm (DP-algorithm) is a classical and well-known  question. However, many enumeration problems have a huge number of solutions in practice, which might be an issue. From a computational point of view, since the number of solutions is a lower bound on the time complexity of any enumeration algorithm, it might make the algorithm impractical on real instances. Furthermore, even if the number of solutions is reasonable enough to be enumerated, the purpose of some applications is to give the output of the algorithm to a human specialist 
(this is necessary, for example, when some of the constraints of the problem are subjective and cannot be modeled).

Indeed, one of the advantages of an enumeration algorithm compared to an optimization one which in general outputs only one optimal solution, is to be able to understand the space of solutions. While this is important in many cases, no human can understand an output composed of billions of solutions. 

The approach generally used to address this consists of enumerating all solutions, and then applying some type of clustering (grouping) algorithm to the set of optimal solutions. The final output presented to the user would then be some ``representative description'' of the clusters (groups) themselves. However, since the number of solutions is a lower bound for the total execution time of any enumeration algorithm, the first step of such a strategy becomes impossible when the number of solutions is too big. A natural question is then whether it would be possible to enumerate directly what we just called a ``representative description'' of the clusters of solutions. This could be for instance an element per cluster. Sometimes a cluster can also be seen as a set of characteristics that the solutions within the cluster share. In such a case, the representative description of a cluster could then be such a set of characteristics. A particularly convenient situation is however when the clusters correspond to equivalence classes of an equivalence relation over the set of solutions that we could establish \emph{a priori}. The output could be in this case the quotient space of the equivalence relation. 

Notice that the enumeration of equivalence classes of solutions is a combinatorial problem that could be solved exactly given a well-defined equivalence relation, and unlike data analysis methods such as incremental clustering, it does not require the definition of a similarity or dissimilarity measure between solutions which, depending on the mathematical nature of the solutions (numerical values, graphs, functions on graphs, etc.), can be difficult to define or costly to compute. 

The problem this paper addresses is how to perform the task of enumerating equivalence classes of solutions with polynomial delay for a wide variety of problems (including optimization problems solved by dynamic programming algorithms), for certain types of equivalence relations between solutions.   

The problem of enumerating equivalence classes, and particularly the generation of representative solutions is a challenge in the context of enumeration algorithms. It has been identified as a need in different areas, such as Genome Rearrangements \cite{Braga2008}, Artificial Intelligence \cite{andersson1997} or Pattern Matching \cite{Blumer87,Narisawa2007}. It was listed as an important open problem in a recent Dagstuhl workshop on ``Algorithmic Enumeration: Output-sensitive, Input-Sensitive, Parameterized, Approximative'' (see, e.g.,  Sections~4.2 and 4.10 in \cite{Dagstuhl}). To the best of our knowledge, this challenge has  been addressed only for some few  specific problems in the literature (e.g., \cite{Angel2011,Braga2008,Molinaro2014,Morrison2015}).

To enumerate equivalence classes, we go through an intermediate problem, namely the enumeration of colored subtrees in  acyclic decomposable AND/OR graphs (ad-AND/OR graph). The paper is organized as follows: \cref{sec:enumeration_AND_OR} provides an algorithm to enumerate with polynomial delay colored subtrees in ad-AND/OR graphs; \cref{section:application} details how this algorithm applies to the enumeration of equivalence classes in DP-problems. In that direction, we present some examples from well-known optimization problems in the literature. Finally, in \cref{sec:conclusions} we conclude with some open problems. 

\section{Enumeration of colored subtrees in an acyclic decomposable AND/OR graph}\label{sec:enumeration_AND_OR}

\subsection{AND/OR graphs and solution subtrees}
An \emph{AND/OR graph} (see, for example \cite{10.1145/359657.359664, Nilsson}) is a well-known structure in the field of Logic and  Artificial Intelligence (AI) that represents problem solving and problem decomposition. In this paper, we consider a particular flavor of AND/OR graphs known as \emph{explicit AND/OR graphs for trees} \cite{DECHTER200773}.

This is a directed acyclic graph (DAG) $G$ which explicitly represents an \emph{AND/OR state space}  for solving a certain problem by decomposing it into subproblems. The set of nodes (or states) $S:=V(G)$ contains OR and AND nodes (the OR nodes represent alternative ways for solving the problem while the AND nodes represent problem decomposition into subproblems, all of which need to be solved). There is a set of goal nodes $S_g\subseteq S$ and a set of start nodes $S_0\subseteq S$ representing respectively the terminal states and the initial states. The children (out-neighbors) of an OR node are AND nodes, and the children of an AND node are OR nodes or goal nodes. We say that a node is an \gor\ node when it is either an OR node or a goal node. Furthermore, the AND/OR graphs that we consider must have the property of being \emph{decomposable} (they can model a problem for which every decomposition yields disjoint subproblems that can be solved independently): for any AND node, the sets of nodes that are reachable from each one of its child nodes are pairwise disjoint. The example graph in \cref{fig:pic1} is decomposable. 

\begin{figure}[ht]
    \begin{center}
    \def\svgwidth{.6\textwidth}
    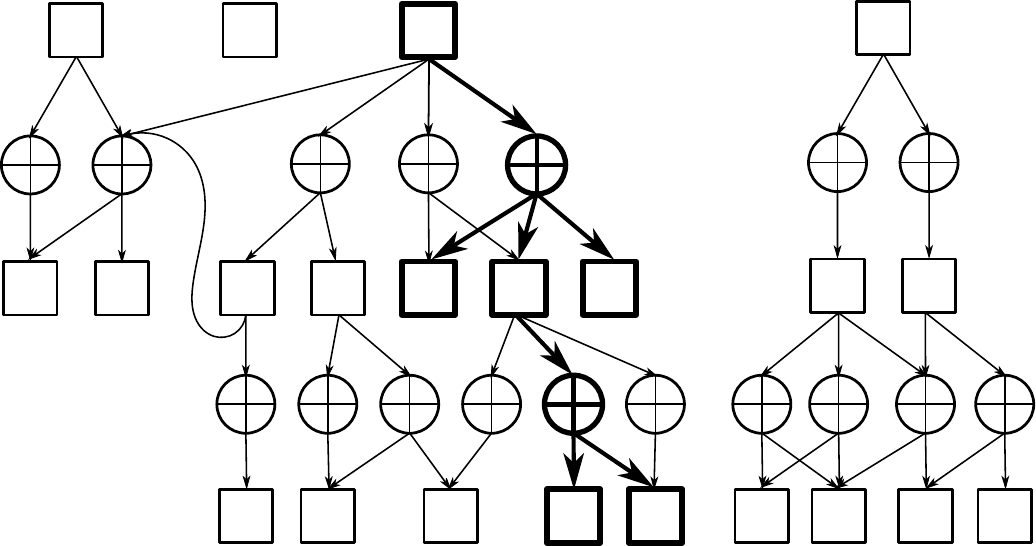
    \caption{An acyclic decomposable AND/OR graph with four start nodes. Squares are \gor\ nodes (OR nodes or goal nodes); crossed circles are AND nodes. One solution subtree of size $8$ is shown in bold. }\label{fig:pic1}
    \end{center}
\end{figure}

Formally, in this paper, any graph that 
satisfies the properties in \cref{defintion:AND-OR-graph} will be called an ad-AND/OR graph. Notice that this definition corresponds only to a particular case of the general AND/OR graphs in the AI literature; the latter may be neither acyclic nor decomposable. 

\begin{definition}[ad-AND/OR graph]\label{defintion:AND-OR-graph}
A directed graph $G$ is an \emph{acyclic decomposable AND/OR graph}, henceforth denoted by ad-AND/OR graph, if it satisfies the following:
\begin{itemize}
\item $G$ is a DAG.
\item $G$ is bipartite: its node set $V(G)$ can be partitioned into $(\mathcal{A}, \mathcal{O})$ so that all arcs of $G$ are between these two sets. Nodes in $\mathcal{A}$ are called \emph{AND nodes}; nodes in $\mathcal{O}$ are called \emph{\gor\ nodes}. 
\item Every AND node has in-degree at least one and out-degree at least one. The set of nodes with out-degree zero is then a subset of $\mathcal{O}$ and is called the set of \emph{goal nodes}; the remaining \gor\ nodes are simply the \emph{OR nodes}. The subset of OR nodes of in-degree zero is the set of \emph{start nodes}.
\item $G$ is decomposable: for any AND node, the sets of nodes that are reachable from each one of its child nodes are pairwise disjoint. 
\end{itemize}
\end{definition}

\begin{definition}[solution subtree]\label{definition:solution-subtree}
A \emph{solution subtree} $T$ of an ad-AND/OR graph $G$ is a subgraph of $G$ which: (1) contains exactly one start node; (2) for any OR node in $T$ it contains one of its child nodes in $G$, and for any AND node in $T$ it contains all its children in $G$. 
\end{definition} 
It is immediate to see that a solution subtree is indeed a subtree of $G$: it is a rooted tree, the root of which is a start node. If we would drop the requirement of $G$ being decomposable, the object defined in \cref{definition:solution-subtree} would not be guaranteed to be a tree.  One solution subtree of the example graph in \cref{fig:pic1} is shown in bold. 

The set of all solution subtrees of $G$ is denoted by $\mathbb{T}(G)$. 
Given an ad-AND/OR graph $G$, counting the number of its solution subtrees and enumerating all solution subtrees can be solved by folklore approaches based on depth-first search (DFS). 

Before going further, we recall that the motivation of this paper is concerned with solutions of dynamic programming problems. The correspondence between the solutions of DP-style recurrence equations and the solution subtrees of general AND/OR graphs has been formally proven in \cite{Gnesi}. 
In the case where the underlying graph is acyclic, the recurrence equations can be solved efficiently by DP-algorithms. 
While we will now concentrate on the solution subtrees of an ad-AND/OR graph and on the equivalence classes of solution subtrees, we will demonstrate in \cref{section:application} how to apply our algorithms to analyze equivalence classes of solutions of a very general class of problems solvable by DP.
It is important to point out that solution subtrees of general AND/OR graphs are equivalent to various other well-known formalisms, e.g., acceptance trees of a nondeterministic tree automata, languages of regular tree grammars, complete subcircuits of tropical circuits. The reader who is more familiar with those may also find this paper interesting even outside of a dynamic programming context. 

\subsection{Equivalence classes}
Let $G$ be an ad-AND/OR graph. Let $C$ be an \textbf{ordered} set of colors. We will consider equivalence relations on the set of solution subtrees of $G$ which are based on a local comparison of the colors of the \gor\ nodes. Intuitively, two \gor nodes having the same color represent two alternative ways of solving the problem that can be considered equivalent.

\begin{definition}[e-coloring]\label{def:colors}
An ad-AND/OR graph $G$ is \emph{e-colored} if its \gor\ nodes are colored in such a way that for any AND node all its children have distinct colors. 
\end{definition}

\subparagraph*{Notations}  If $s$ is a \gor node of $G$, we denote by $c(s)$ its color. If $s$ is an AND node, we denote by $\widetilde{C}(s)$ the tuple of colors of the children of $s$ sorted in increasing order of the colors.
If $T_1\in\mathbb{T}(G)$ is a solution subtree of $G$, we use the notation $\pi(T_1)$ for the result of contracting the AND nodes in $T_1$: for each OR node $s$ of $T$, contract the only child node of $s$ in $T$ (i.e., remove the child and connect $s$ to each one of its ``grandchildren''). 

\begin{definition}[equivalence class]\label{definition:equivalence-classes}
A node-colored rooted tree $T$ is an \emph{equivalence class} of solution subtrees of an e-colored ad-AND/OR graph $G$ (or shortly an equivalence class of $G$) if there exists a solution subtree $T_1$ of $G$  such that $\pi(T_1)$ is equal to $T$. Such a $T_1$ is said to be a solution subtree \emph{belonging to} the class $T$.
\end{definition}

\subparagraph*{More notations} We denote by $\mathbb{C}(G)$ the set of equivalence classes of $G$. The notation $\pi$ can be seen as a function $\pi\colon \mathbb{T}(G)\to \mathbb{C}(G)$. We denote by $\pi^{-1}(T):=\{T_1\in\mathbb{T}(G)\mid \pi(T_1)=T\}$ the subset of solution subtrees of $G$ belonging to the class $T$. The notations $c(s)$ and $\widetilde{C}(s)$ are naturally extended to the case where $s$ is a node in an equivalence class $T$. The root node of a rooted tree $T$ is denoted by $r(T)$. The set of the children of a node $s$ is denoted by $Ch(s)$. 
 
An example of an e-colored ad-AND/OR graph with five equivalence classes is given in \cref{fig:pic2}.

\begin{figure}[ht]
    \begin{center}
    \def\svgwidth{.7\textwidth}
    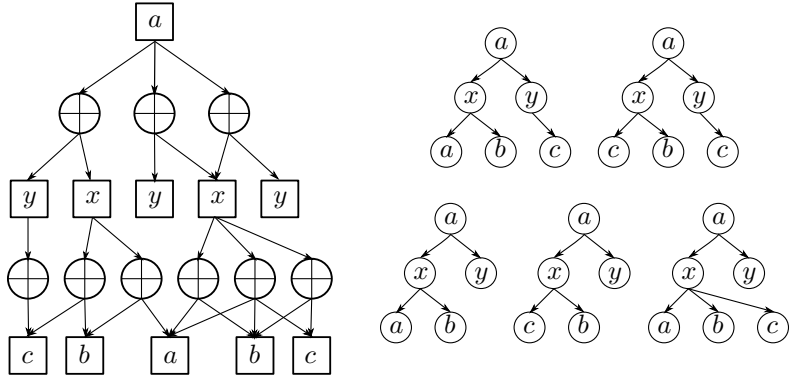
    \caption{An e-colored ad-AND/OR graph and its five equivalence classes. The colors of the \gor\ nodes are written inside the squares.}\label{fig:pic2}
    \end{center}
\end{figure}

\subsection{Enumerating equivalence classes}

Given an e-colored ad-AND/OR graph $G$, we propose a polynomial delay algorithm to enumerate all equivalence classes of $G$. Given a total ordering $c_1,\dots,c_m$ of the colors of $G$, we define a total ordering $\prec$ over $\mathbb{C}(G)$, the set of equivalence classes of $G$. If $T$ and $T'$ have their roots colored differently, we say that $T$ is smaller than $T'$, denoted by $T\prec T'$, if the root color of $T$ precedes the one of $T'$.
If $T$ and $T'$ have the same root color, let $(T_1,\dots,T_k)$ (resp. $(T'_1,\dots,T'_\ell)$) be the child subtrees of $r(T)$ (resp. $r(T')$) sorted recursively with respect to $\prec$. We then say that $T$ is smaller than $T'$ if the tuple $(T_1,\dots,T_k)$ is lexicographically smaller than $(T'_1,\dots,T'_\ell)$, i.e., if $T_i\prec T'_i$ with $i$ being the smallest index such that $T_i \neq T'_i$. We also assume that $\emptyset$ is smaller than any tree, and 
therefore a single node tree colored with color $c$ comes before any other tree whose root is colored with $c$ in $\prec$. 

\begin{algorithm}[!htb]
\DontPrintSemicolon
  \SetKw{Yield}{Return}
  \SetKw{Input}{Input:}
  \SetKw{Output}{Output:}
  \SetKw{GoTo}{Go To}
  \SetKwFunction{FNext}{Next}
  \SetKwFunction{FNewTree}{NewTree}
  \SetKwProg{Fn}{Function}{:}{}
  \Input A set $\cO$ of \gor  nodes having all the same color $c$ and an equivalence class $T$ of $G/\cO$\;
  \Output The equivalence class $T'$ of $G/\cO$ that follows $T$ w.r.t the $\prec$ ordering.

  \Fn{\FNext{$T$,\cO}}{
  
\If{$T=\emptyset$ and $\cO$ contains terminal nodes}
{
            \Yield A tree with a single root node colored with $c$\; 
}
  $r\gets 0$\;
\eIf{$T = \emptyset$ or $T$ is a single node tree }{\label{begining}
        $r\gets r+1$ \;
         \If{$r>|\cT(\cO)|$}
        {
          \Yield $\bot$ 
        }
        Let $(c_1,\dots,c_j)$ be the color tuple of $t_r \in \cT(\cO)$ and let $T_1,\dots,T_j\gets \emptyset$\;
        $\cO_1 \gets C^r_1$\;
        
        $\ell \gets 1$\;
  }
  {
        Let $r$ be such that $t_r \in \cT(\cO)$ is the root color tuple $\widetilde{C}(r(T))$\;
        
        Let $(T_1,\dots,T_j)$ be the child subtrees of $r(T)$, the 
        roots of which are colored respectively with $(c_1,\dots,c_j):=t_r$ \;

        For all $i\leq j$, let $\cO_i\subseteq C^r_i$  be the set of nodes in $C^r_i$ compatible with $(T_1,\dots,T_{i-1})$\;

        Let $\ell$ be the largest index $i\leq j$ such that \FNext{$T_i$,$\cO_i$} $\neq \bot$ if such index exists. Otherwise, $T\gets \emptyset$ and go to line \ref{begining}\;
  }

\BlankLine
  
        $T_\ell \gets$ \FNext{$T_\ell$,$\cO_\ell$}\;
        \For{$\ell < i \leq j$}
        {
            Let $\cO_i \subseteq C^r_i$  be the set of nodes in $C^r_i$ compatible with $(T_1,\dots,T_{i-1})$\;
            $T_i \gets$ \FNext{$\emptyset$,$\cO_i$}\;
        }
        \Yield A tree with root color $c$ and root child subtrees $(T_1,\dots,T_j)$\;
}

\caption{Next solution}

\label{alg:Newone}
\end{algorithm}

\subsubsection{Definitions and notations}

Recall that given an AND-node $x$, $\widetilde{C}(x)$ is the tuple of colors of the children of $x$ sorted in increasing order.
Given an  OR-node $o$, we denote by $\cT(o)$ the set of color tuples of its children, i.e., $\cT(o):=\{\widetilde{C}(x) : x\in Ch(o) \}$. In other words, a color tuple $(c_1,\dots,c_j)$ belongs to $\cT(o)$ if $o$ has an AND-child node whose children are colored with $(c_1,\dots,c_j)$.
If we consider an equivalence class $T$ of $\cS(G/\{o\})$ rooted at $o$, the tuples of $\cT(o)$ are precisely the possible colorings of the children of $r(T)$. Indeed, if the AND-child node $x\in Ch(o)$ is chosen in a solution subtree, then $\widetilde{C}(x)$ will be the colors of the children of $o$ in that solution.
Notice that several AND-children nodes of $o$ may have the same color tuple.

We extend this definition to a set $\cO$ of $\gor$ nodes with $\cT(\cO)=\bigcup_{o\in \cO} \cT(o)$. 
In the same way, $t=(c_1,\dots,c_j)$ is a color tuple of $\cT(\cO)$ if and only if there exists an equivalence class $T$ of $\cS(G/\cO)$ such that the children of $r(T)$ are colored with $(c_1,\dots,c_j)$. 
Given a set $\cO$ of $\gor$ nodes, we denote by $t_1,\dots,t_{|\cT(\cO)|}$ the different color tuples of $\cT(\cO)$ ordered lexicographically, and we denote by $Ch^{\ell}(\cO)$ the set of AND-nodes in $Ch(\cO)$ whose color tuple is $t_\ell$, i.e., $Ch^{\ell}(\cO)=\{x\in Ch(\cO): \widetilde{C}(x) = t_\ell\}$.
The sets $Ch^{1}(\cO),\dots,Ch^{|\cT(\cO)|}(\cO)$ form a partition of $Ch(\cO)$, each part corresponding to a color tuple $t_i \in \cT(\cO)$.

Finally, given a color tuple $t_\ell :=(c_1,\dots,c_j)\in \cT(\cO)$, for each $i\leq j$, by \cref{def:colors}, each node of $Ch^\ell(\cO)$ has exactly one child with color $c_i$. We denote by $C^{\ell}_i$ the set of children of $Ch^{\ell}(\cO)$ colored with  $c_i$, i.e., $C^{\ell}_i=\{o\in Ch(x):x\in Ch^\ell(\cO),\; c(o)=c_i \}$ (it is a set of ``grandchildren'' of $\cO$). 

In the left panel of \cref{fig:pic3}, an example graph is shown where each node is labeled by an integer. The colors are, in increasing order, $w$, $x$, $y$, and $z$. For $\cO=\{1,2\}$, the set $\cT(\cO)$ contains the three tuples $t_1=(w,y)$, $t_2=(x,y)$, and $t_3=(x,y,z)$. We have   $Ch^1(\cO)= \{6\}$, $Ch^2(\cO)=\{4,5\}$, and $Ch^3(\cO)=\{3\}$. 
    The sets $C^{\ell}_i$ are $C^{1}_1= \{12\}$, $C^{1}_2= \{11\}$, $C^{2}_1= \{9,10\}$, $C^{2}_2=\{8, 11\}$, $C^{3}_1=\{9\}$, $C^{3}_2=\{8\}$, and $C^{3}_2=\{7\}$. 

\begin{figure}
    \centering
    \def\svgwidth{.8\textwidth}
    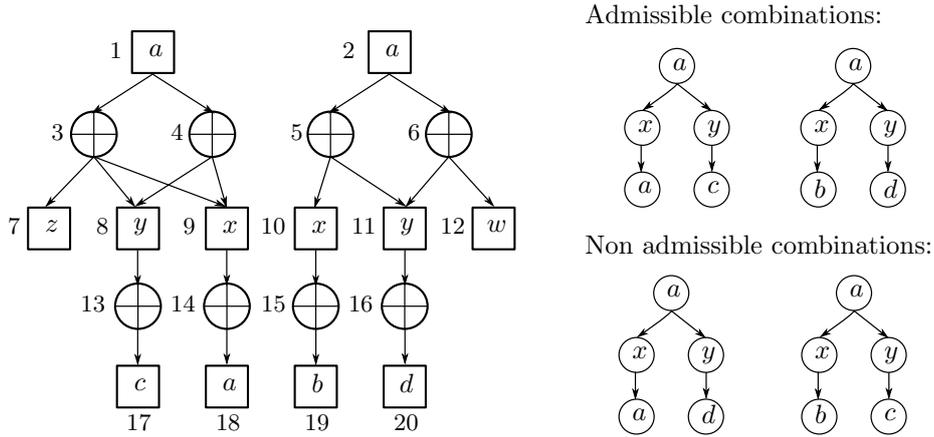
    \caption{Left panel: An e-colored ad-AND/OR graph. Right panel: For $\cO=\{1,2\}$, there are four combinations between $\cS(G/C_1^2)$ and $\cS(G/C_2^2)$; only two of them are admissible.}
    \label{fig:pic3}
\end{figure}

\subsubsection{Algorithm description}

Notice that by definition of $\prec$, given a set of \gor nodes $\cO$, all having the same color $c$, and a color tuple $t_\ell\in \cT(\cO)$, all equivalence classes $T$ of $G/\cO$ such that $\widetilde{C}(r(T))=t_\ell$ are consecutive with respect to $\prec$.

The algorithm outputs the equivalence classes in ascending order with respect to $\prec$. Given an equivalence class $T$ of $G/\cO$ for a set of \gor nodes $\cO$ of color $c$, it will output the equivalence class $T'$ of $G/\cO$ that succeeds $T$ w.r.t. $\prec$ if it exists or output the symbol $\bot$ if $T$ is the last solution. 

Assume that the children of $r(T)$ are colored with the \emph{root color tuple} $(c_1,\dots,c_j)=:t_r\in \cT(\cO)$ and let $(T_1,\dots,T_j)$ be the child subtrees of $T$, the roots of which are colored with the tuple $t_r$. Notice that for all $i\leq j$, $T_i$ is an equivalence class of $G/C^r_i$.
The algorithm will output the next equivalence class $T'$ such that $\widetilde{C}(r(T'))=t_r$ if there remains one (same color tuple at the root), or it will output the first solution such that $\widetilde{C}(r(T'))=t_{r+1}\in \cT(\cO)$ otherwise (the next color tuple at the root).

To find the next solution corresponding to the root color tuple $t_r$, the algorithm will replace recursively $T_j$ by its successor $T'_j$ w.r.t. $\prec$ if there exists one. We obtain the solution $T'$ whose subtrees are $(T_1,\dots,T_{j-1},T'_j)$ which is by definition the successor of $T$ in $\prec$ whenever $T'_j$ is the successor of $T_j$. If $T_j$ has no successor (that is, if it is the last one), we replace if possible $T_{j-1}$ by its successor $T'_{j-1}$ and we replace $T_j$ by the smallest admissible solution (i.e., the successor of $\emptyset$). In general, we select at each step the greatest index $\ell$ such that $T_\ell$ has a successor w.r.t. $\prec$, we replace it by its successor $T'_\ell$ and we take the smallest admissible solution for every $\ell <i\leq j$. 

Without further care, the above described procedure would output solutions whose child subtrees $(T_1,\dots,T_j)$ of the root correspond to the elements of the Cartesian product of $\cS(G/C^r_i)$, $i\leq j$. However, while it is true that if $T$ is a solution, its child subtree $T_i$ is an equivalence class of $G/C^r_i$ for all $i\leq j$, the converse is not true. Indeed, not all elements of $\cS(G/C^r_1)\times\dots\times\cS(G/C^r_j)$ 
lead to an admissible solution 
(an example is given in the right panel of \cref{fig:pic3}).
In order to find an admissible solution, we should guarantee that the choice of a given $T_i$ is \emph{compatible} with the previous choices $(T_1,\dots,T_{i-1})$. This is done by selecting the subset of \gor nodes $\cO_i \subseteq C^r_i$ that are compatible with $(T_1,\dots,T_{i-1})$ (see \cref{defn:compatible} below). An admissible choice 
of  $T_i$ will then be any equivalence class of $G/\cO_i$. The two key properties are that the set $\cO_i$ can be easily computed, and that it is never empty, i.e., there is always a choice for $T_i$ that is compatible with the previous choices of $(T_1,\dots,T_{i-1})$ (there is at least one choice that corresponds to the current solution). Notice that if the latter were not true, the algorithm would not have a polynomial delay complexity since we may spend exponential time without reaching a final solution. With this property, we are guaranteed that we can always extend a partial tuple $(T_1,\dots,T_{i})$ until we reach a complete tuple $(T_1,\dots,T_{j})$ that will form a solution.

\paragraph*{Compatible nodes}

Given a set of \gor nodes $\cO$ all colored with the same color $c$ and a tree $T$ of 
$\cS(G/\cO)$, we denote by $r(\pi^{-1}(T)):=\{r(S):S\in \pi^{-1}(T)\}$ the subset of \gor nodes of $\cO$, each one of which is the root of a solution subtree of class $T$. The following definition formalizes the notion of compatible nodes mentioned previously.  

\begin{definition}\label{defn:compatible}
Let $\cO$ be a set of \gor nodes of color $c$, $t_r=:(c_1,\dots,c_j)\in \cT(\cO)$, and let $T_1,\dots,T_k$, with $k<j$, be respectively 
equivalence classes of $\cS(G/C^r_i)$ for all $i\leq k$. We say that a node $o\in C^r_{k+1}$ is \emph{compatible} with $(T_1,\dots,T_k)$ if there exists an AND-node $x\in Ch^r(\cO)$ such that $o$ is a child of $x$ and such that 
$r(\pi^{-1}(T_i))$ contains a child of $x$ for all $i\leq k$. 
\end{definition}

\subsubsection{Analysis}

\begin{lemma}\label{lemma:compatible_sets}
Let $\cO$ be a set of \gor nodes of color $c$, $t_r=:(c_1,\dots,c_j)\in \cT(\cO)$. Let $T\in \cS(G/\cO)$, let $T_1,\dots,T_k$, $k<j$ be its first $k$ child subtrees with $T_i \in \cS(G/C^r_i)$ for all $i\leq k$ and let $\cO_{k+1}\subseteq C^r_{k+1}$ be the set of \gor nodes compatible with $(T_1,\dots,T_k)$. Given $T_{k+1} \in \cS(G/C^r_{k+1})$, there exists a tree  $T'\in \cS(G/\cO)$ whose first $k+1$ child subtrees are $(T_1,\dots,T_k,T_{k+1})$ if and only if $T_{k+1}\in \cS(G/\cO_{k+1})$. 
\end{lemma}

\begin{proof}
$(\Rightarrow)$ Assume that there exists a tree  $T'\in \cS(G/\cO)$ whose first $k+1$ child subtrees are  $(T_1,\dots,T_k,T_{k+1})$, and let $C$ be a solution subtree of $G/\cO$ such that $\pi(C)=T'$. Let $x$ be the AND-child node of the root of $C$. Notice that since we assumed that $T_i\in C^r_i$ for all $i\leq k$, $\widetilde{C}(x)=t_r$, and so $x\in Ch^r(\cO)$. Let $(o_1,\dots,o_j)$ be the children of $x$. Since for all $i\leq k$ we have $o_i \in r(\pi^{-1}(T_i)$, the node $o_{k+1}$ is compatible with $(T_1,\dots,T_k)$. Thus $T_{k+1} \in \cS(G/\cO_{k+1})$ since $T_{k+1}\in \cS(G/\{o_{k+1}\})$, and $o_{k+1} \in \cO_{k+1}$.

$(\Leftarrow)$ Assume now that $T_{k+1}\in \cS(G/\cO_{k+1})$. There exists $o_{k+1} \in \cO_{k+1}$ and a solution subtree $C_{k+1}$ of $G/\{o_{k+1}\}$ rooted at $o_{k+1}$ with $\pi(C_{k+1})=T_{k+1}$. Since $o_{k+1}\in \cO_{k+1}$, there exists an AND-child node $x$ of a node $o$ in $\cO$ such that $\widetilde{C}(x)=t_r$, and $o_i \in r(\pi^{-1}(T_i))$ for all $i\leq k$ where $o_i$ is the unique child of $x$ of color $c_i$. Therefore, for all $i\leq k$, there exists a solution subtree $C_i$ of $G/\{o_i\}$ such that $\pi(C_i)=T_i$. 
Now consider any solution subtree $C$ of $G/\cO$ rooted at $o$, with $o$ having $x$ as AND-child node and with $x$ having its first $k+1$ child solution subtrees equal to $(C_1,\dots,C_{k+1})$.
Then the first $k+1$ child subtrees of $\pi(C)$ will be $(T_1,\dots,T_k,T_{k+1})$.      
\end{proof}

\begin{proposition}\label{lemma:correctness}
Let $T$ be an equivalence class of $G/\cO$ for a set $\cO$ of \gor nodes of $G$, all colored with the same color $c$. Then,  the function \texttt{Next} of \cref{alg:Newone} is such that:
\begin{enumerate}
    \item \texttt{Next}($\emptyset,\cO$) 
    returns the smallest equivalence class of $G/\cO$ w.r.t. $\prec$.
    \item \texttt{Next}($T,\cO$) 
    returns the equivalence class of $G/\cO$ that 
    follows $T$ w.r.t. $\prec$.
    \item if $T$ is the last equivalence  class of $G/\cO$,
    \texttt{Next}($T,\cO$) 
    returns $\bot$.
\end{enumerate}
\end{proposition}

\begin{proof}
Let us define the height of $G/\cO$, 
$h(G/\cO)$  to be the maximum height of an equivalence class of $G/\cO$, i.e., the number of \gor nodes in a longest path from $\cO$ to a goal node minus $1$. The proof will be done by induction on $h(G/\cO)$.

Assume first that $h(G/\cO)=0$, i.e., $\cO$ contains only goal nodes. Then $G/\cO$ has only one equivalence class $T$, which is the single node tree of color $c$. The call of \texttt{Next}($\emptyset,\cO$) will output it in Line~5 of the algorithm, and the call of \texttt{Next}($T,\cO$) will return  $\bot$ in 
Line 11 since $|\cT(\cO)|=0$.

Assume now that $h(G/\cO)>0$.

\subparagraph*{Proof of \ref{lemma:correctness}.1} If $\cO$ contains goal nodes, then the smallest equivalence class of $G/\cO$ with respect to $\prec$ is the single node tree colored with $c$ and \texttt{Next}($\emptyset,\cO$) outputs it in Line 5. 
Otherwise, let $T$ be the smallest equivalence class of $G/\cO$ and let  $(T'_1,\dots,T'_j)$ be the subtrees of $T$ rooted at the children of $r(T)$.
By definition of $\prec$, the roots of $(T'_1,\dots,T'_j)$ are colored with the minimum color tuple $t_1:=(c_1,\dots,c_j)$ of $\cT(\cO)$ and by 
\cref{lemma:compatible_sets}, for all $i\leq j$ $T'_i$ is the smallest equivalence class of $G/\cO_i$ where $\cO_i \subseteq C_i^1$ is the set of nodes of $C_i^1$ compatible with $(T'_1,\dots,T'_{i-1})$.  Thus, \texttt{Next}($\emptyset,\cO$) will return $T$ in Line~27 since $r$ will receive $1$ in Line~9, $T_1$ will receive \texttt{Next}($\emptyset,C_1^1)$
in Line~22 which is equal to $T'_1$ by 
the induction hypothesis, and for all $i\leq j$, $T_i$ will receive \texttt{Next}($\emptyset,\cO_i$) in Line~25 which is equal to $T'_i$ by 
the induction hypothesis.

\subparagraph*{Proof of \ref{lemma:correctness}.2}
Let $(T'_1,\dots,T'_j)$ be the child subtrees of $r(T)$ and let $t_r:=(c_1,\dots,c_j)$ be the color tuple of $\cT(\cO)$ with which their roots are colored. Let $T'$ be the equivalence class of $G/\cO$ that follows $T$ with respect to $\prec$. Notice that the children of $r(T')$ are either colored with $t_r$ or with $t_{r +1}$ if $T$ is the largest equivalence class whose root children are colored with $t_r$.

Assume first that the children of $r(T')$ are colored with the color tuple $t_r$. Let  $(T''_1,\dots,T''_j)$ be the child subtrees of 
$r(T')$ and let $\cO_i \subseteq C_i^r$ be the set of nodes of $C_i^r$ compatible with $(T'_1,\dots,T'_{i-1})$ for all $i\leq j$. Let $\ell$ be the smallest index such that $T'_\ell \neq T''_\ell$.   We claim that $\ell$ is also the largest index such that $T'_\ell$ has a successor in $\cS(G/\cO_\ell)$ with respect to 
$\prec$, and thus that it corresponds to the $\ell$ chosen by the algorithm in Line 20. 
Indeed, assume that there exist $i<k\leq j$ and $F\in \cS(G/\cO_k)$ such that $T'_k\prec F$. By Lemma \ref{lemma:compatible_sets}, there exists an equivalence class of $G/\cO$ whose first $k$ child subtrees would be $(T'_1,\dots,T'_{k-1},F)$. 
However, in this case such an equivalence class would be greater than $T$ and smaller than $T'$ with respect to $\prec$, 
and it would be in contradiction with the fact that $T'$ immediately follows $T$ in $\prec$.
Now $T_\ell$ will receive  \texttt{Next}($T'_\ell,\cO_\ell$) 
in Line 22 which is by the induction hypothesis the tree of $\cS(G/\cO_\ell)$ that follows $T'_\ell$. Since we assumed that $\ell$ is the smallest index such that $T'_\ell \neq T''_\ell $, by definition of $\prec$ and by \cref{lemma:compatible_sets}, $T''_\ell$ is the equivalence class of $G/\cO_\ell$ that follows $T'_\ell$ in $\prec$, and so $T_\ell$ will receive $T''_\ell$ in Line 22. 
Since for all $i\leq \ell$ $T'_i = T''_i$, and since $(T_1,\dots,T_{\ell-1})$ are not modified by the algorithm, at the end of it, $(T_1,\dots,T_{\ell-1},T_\ell)$ will be equal to $(T'_1,\dots,T'_{\ell-1},T''_\ell)=(T''_1,\dots,T''_{\ell-1},T''_\ell)$. 
It now remains to show that $T_i$ will be equal to $T''_i$ for all $\ell< i \leq j$. 
Again, by \cref{lemma:compatible_sets}, for all $\ell<i\leq j$, $T''_i$ is the smallest tree of $\cS(G/\cO'_i)$ where $\cO'_i$ is the set of nodes of $C^r_i$ compatible with $(T''_1,\dots,T''_{i-1})$ since
otherwise, another tree of $\cS(G/\cO)$  greater than $T$ and smaller than $T'$ could be built. Thus, in Line~25, $T_i$ will receive \texttt{Next}($\emptyset,\cO'_i$) which is equal to $T''_i$ by the induction hypothesis, and $T'$ will be returned in Line~27.

Assume now that the children of $r(T')$ are colored with the color tuple $t_{r+1}$. In this case, $T$ is the greatest tree of $\cS(G/\cO)$ with respect to $\prec$ whose root children are colored with $t_r$.  Let $\cO_i \subseteq C_i^r$ be the set of nodes of $C_i^r$ compatible with $(T'_1,\dots,T'_{i-1})$ for all $i\leq j$. We claim that $T'_{i}$ is the greatest tree of $\cS(G/\cO_i)$ for all $i\leq j$ with respect to $\prec$. Indeed, assume otherwise that there exists $T''_i\in \cS(G/\cO_i)$ with $T'_i \prec T''_i$. Then by \cref{lemma:compatible_sets} $\cS(G/\cO)$ would contain a tree whose root children would be colored with $t_r$ which would be larger than $T$ with respect to $\prec$. By the induction hypotheses, \texttt{Next}($T'_i,\cO_i$) will therefore return $\bot$ for all $i\leq j$. So $T$ will receive $\emptyset$ in Line 20 and the algorithm will jump to Line~8, and the next color tuple $t_{r+1}$ will be selected in Line~9. Using now similar arguments to the ones used in the proof of \ref{lemma:correctness}.1, the smallest tree with respect to $\prec$ whose root children are colored with $t_{r+1}$ will be returned. 

\subparagraph*{Proof of \ref{lemma:correctness}.3}
Assume now that $T$ is the last equivalence class of $G/\cO$. Notice that in this case, the root children of $T$ are colored with $t_r$ where $r=|\cT(\cO)|$. As previously, $T'_{i}$ is the greatest tree of $\cS(G/\cO_i)$ for all $i\leq j$ with respect to $\prec$ since otherwise a greater tree of $\cS(G/\cO)$ would exist. By the induction hypothesis, \texttt{Next}($T'_i,\cO_i$) would return $\bot$ for all $i\leq j$. Therefore $T$ will receive $\emptyset$ in Line~20 and the algorithm will return to Line~8. Since $r=|\cT(\cO)|$, $r$ will receive $r+1$ in Line~9 and $\bot$ will be returned in Line~11.

\end{proof}

\begin{theorem}
Given an e-colored ad-AND/OR graph $G$, the set $\cS(G)$ can be enumerated with delay $O(n\cdot s)$ where $n$ is the number of
nodes of $G$ and $s$ is the maximum size of a solution. 
\end{theorem}
\begin{proof}
To enumerate $\cS(G)$, we first split the 
start nodes of $G$ into sets $S_0,\dots,S_k$ according to their colors. For each set $S_i$, starting with $T=\emptyset$, we repeatedly assign \texttt{Next}($T,S_i$) to  $T$ and output it until $T=\bot$. By \cref{lemma:correctness}, 
this guarantees that we output every solution of $\cS(G/S_i)$ exactly once. Since any solution of $\cS(G)$ belongs to $\cS(G/S_i)$ for a given $i\leq k$, every solution of $\cS(G)$ will be outputted exactly once.

For the complexity, notice that at most one recursive call is performed by the node of the next solution. More precisely, if \texttt{Next}($T,\cO$)$=T'$, 
there will be exactly one recursive call per node in $T'$ that is not in $T$, and thus at most $s$ recursive calls will be performed. 

In each recursive call, both the set $\cT(\cO)$ and the partition $\{C^r_i\}_{i\leq j}$ of grandchildren of $\cO$ can be computed in $O(n)$
time. It remains to show that the sets of compatible nodes $\cO_i$, $i\leq j$, can be computed in $O(n)$ 
time in total which will conclude the proof. 
To do this, we should be able to compute the sets $r(\pi^{-1}(T_i))$ for all $i\leq j$. If \texttt{Next}($T,\cO$)$=T'$, the easiest way is to return the set $r(\pi^{-1}(T'))$ together with $T'$ when the call \texttt{Next}($T,\cO$) returns. This could be done by observing that if $T\in\cS(G/\cO)$ where $\cO$ is a set of goal nodes all having the same color $c$, then  $r(\pi^{-1}(T))=\cO$, and if $T$ has child subtrees $T_1,\dots,T_j$ then $r(\pi^{-1}(T))$ is the set of nodes of $\cO$ that has at least an AND-child $x$ such that the children of $x$ contain exactly one node in $r(\pi^{-1}(T_i))$ for each $i\leq j$, which can be found in $O(n)$ time. Thus only $O(n)$ time is necessary at each recursive call to return $r(\pi^{-1}(T'))$ in addition to $T'$. 
\end{proof}

\subsection{Restricting the graph to an equivalence class}

After the enumeration of the equivalence classes, it might be interesting to go back to the solutions in each class, in particular, one might want to use the number of solutions as a measure for the ``importance'' or ``significance'' of an equivalence class. We present an algorithm that, given an e-colored ad-AND/OR graph $G$ and an equivalence class $T$, constructs the \emph{subgraph $G^T$ of $G$ restricted to $T$}, that is, a subgraph of $G$ of which the solution subtrees are exactly the ones of $G$ belonging to the equivalence class $T$:  $\mathbb{T}(G^T)=\pi_G^{-1}(T):=\{T_1\in\mathbb{T}(G)\mid \pi(T_1)=T\}$. 
Once the graph $G^T$ is obtained, the following questions can be answered (by applying the same method as for the unrestricted ad-AND/OR graph $G$): counting the number of, and enumerating the solution subtrees belonging to the class $T$.

\cref{alg:2} relies on two recursive functions \texttt{VisitOR} and \texttt{VisitAND}, both taking as input a node in $G$ and a node in $T$. The \textbf{Require} statements are used to specify the preconditions that the two parameters of the two \texttt{Visit} functions must verify; it can be checked by inspection that these conditions are always satisfied whenever the functions are called. The algorithm performs an operation called \textsf{Mark} on the nodes in $G$. All nodes are initially unmarked; the \textsf{Mark} operation changes the state of a node into marked.

\begin{algorithm}[!htb]
\DontPrintSemicolon
  \SetKw{Yield}{yield}
  \SetKw{Require}{Require:}
  \SetKwFunction{FMain}{Main}
  \SetKwFunction{FVisitOR}{VisitOR}
  \SetKwFunction{FVisitAND}{VisitAND}
  \SetKwProg{Fn}{Function}{:}{}
  
  \KwData{an e-colored ad-AND/OR graph $G$, an equivalence class $T$}
  \KwResult{the graph $G^T$}
  \Fn{\FMain{$G$, $T$}}{
    \For{each start node $s_0$ of $G$ such that $c(s_0)=c(r(T))$}{
     \FVisitOR{$s_0$, $r(T)$}
    }
    \KwRet $G^T$ obtained from $G$ by removing all unmarked nodes
  }
\BlankLine

  \Fn{\FVisitOR{$s$, $v$}}{
  \Require $c(s)=c(v)$\;
    \If{$s$ is a goal node and $v$ is a leaf}{
        \textsf{Mark}($s$)\;
        \KwRet\;
    }
    \For{each child node $s_i$ of $s$ in $G$ such that $\widetilde{C}(s_i)=\widetilde{C}(v)$}
    {    \FVisitAND($s_i$, $v$)\;
    }
    \If{at least one child of $s$ is marked}{
        \textsf{Mark}($s$)\;
    }
  }
\BlankLine
  \Fn{\FVisitAND{$s$, $v$}}{
    \Require $\widetilde{C}(s)=\widetilde{C}(v)$\;
    \For {each child node $s_i$ of $s$} 
    {
        $v_i\gets$ the unique child of $v$ such that $c(v_i)=c(s_i)$\;
        \FVisitOR($s_i$, $v_i$)
    }
    \If{all children of $s$ are marked}{
        \textsf{Mark}($s$)\;
    }
  }

\caption{Restricting the graph to an equivalence class}
\label{alg:2}
\end{algorithm}

Recall that $\pi(T_1)$ transforms a solution subtree $T_1\in\mathbb{T}(G)$ into an equivalence class $T\in\mathbb{C}(G)$ by contracting the AND nodes in $T_1$. For a fixed $T_1$, we extend this notation and write $\pi(s)=\pi(s_1):=v$ for every OR node $s$ in $T_1$ with its unique AND-child node $s_1$ in $T_1$ that are identified with the node $v$ in $T$ under the transformation.
 
\begin{lemma}\label{lem:2}
Let $T_1$ be a solution subtree of $G$ belonging to the class $T$. For every node $s$ of $T_1$, there is a call to either \emph{\texttt{VisitOR}} or 
to \emph{\texttt{VisitAND}} of \cref{alg:2} with parameters $s$ and $\pi(s)$, depending on whether $s$ is an \gor\ node or an AND node. 
\end{lemma}

\begin{proof} By top-down induction. The start node of $T_1$ is visited in  Line~3 since it has the correct color. In the induction step we separate two cases. For an \gor\ node $s$ of $T_1$ that is not a start node, suppose that the parent $s_1$ of $s$ in $T_1$ is visited in a call \texttt{VisitAND}($s_1$, $\pi(s_1)$). Then $s$ is visited (Line~22), and the second parameter is $\pi(s)$. In the other case, for an AND node $s$ of $T_1$, suppose that the parent $s_1$ of $s$ in $T_1$ is visited in a call \texttt{VisitOR}($s_1$, $\pi(s_1)$). Since we have $\pi(s)=\pi(s_1)$ and $\widetilde{C}(\pi(s_1)) = \widetilde{C}(s)$, the
condition at Line~12 is satisfied and $s$ is visited in a call \texttt{VisitAND}($s$, $\pi(s)$).
\end{proof}

The correctness of \cref{alg:2} is shown in \cref{proposition:correctness2}. We omit the analysis of complexity as the algorithm clearly requires a running time that is linear in the size of the graph.
\begin{proposition}\label{proposition:correctness2}
The set of solution subtrees of $G^T$, the graph returned by \cref{alg:2}, is equal to $\pi_G^{-1}(T)$, i.e., the set of solution subtrees of $G$ belonging to the equivalence class $T$.
\end{proposition}

\begin{proof}
\emph{(First direction)} 
We show that any solution subtree of $G^T$ is also a solution subtree of $G$, and that it belongs to the equivalence class $T$. For every marked OR node, at least one child is marked (Line~15); for every marked AND node, all its children are marked (Line~24). A solution subtree of $G^T$ is thus a solution subtree of $G$. Let $T_1$ be a solution subtree of $G^T$, consider the recursion tree of the Visit function calls during which the nodes in $T_1$ are marked. By the preconditions of the Visit functions ($c(s) = c(v)$ and $\widetilde{C}(s) = \widetilde{C}(v)$), the tree $\pi(T_1)$ is equal to the tree formed by the colored nodes that are used as the second parameter $v$ in this recursion tree. The latter is simply equal to $T$ (we start with the root of $T$, then visit each child of the current node), so we have $\pi(T_1) = T$. Therefore, every solution subtree of $G^T$ belongs to the class $T$.

\emph{(Second direction)} 
Let $T_1$ be a solution subtree of $G$ such that $\pi(T_1)=T$, we show that every node in $T_1$ is marked by bottom-up induction. By \cref{lem:2}, any goal node $s$ in $T_1$ is visited in a call \texttt{VisitOR}($s$, $\pi(s)$) so $s$ is marked (Line~9) because $\pi(s)$ is necessarily a leaf. For the induction step we separate two cases. Let $s$ be an OR node in $T_1$, and suppose that all nodes in $T_1$ at a smaller height are marked. By the same lemma, $s$ is visited. Then $s$ is marked at Line~16, since exactly one child of $s$ is in $T_1$ and is thus marked. In the other case, let $s$ be an AND node in $T_1$ and suppose that all nodes in $T_1$ at a 
smaller height are marked. By the lemma, $s$ is visited. Then $s$ is marked at Line~25 because all children of $s$ are in $T_1$ and are thus marked. This completes the proof.
\end{proof}

\section{Application to dynamic programming}\label{section:application}

\subsection{A formalism for tree-sequential dynamic programming}
Since its introduction by Karp and Held \cite{Karp67}, \emph{monotone sequential decision processes} (mSDP) have been the classical model for problems solvable by dynamic programming (DP). This formalism is based on finite-state automata. The solutions of DP-problems are thus equivalent to languages of regular expressions, or to paths in directed graphs. It is known that Bellman's \emph{principle of optimality} \cite{bellman2013dynamic} also applies to problems for which the solutions are not sequential but \emph{tree-like} \cite{Bonzon1970NecessaryAS}. Various generalizations have been proposed to characterize broader classes of problems solvable by DP or DP-like techniques \cite{BOSI, Helman89}. In this paper, we consider a framework which is the immediate generalization of the mSDP model, i.e., generalizing  finite automata (regular expressions, paths in DAGs) to finite tree automata (regular tree grammars, solution trees of general AND/OR graphs). Further generalizations exist (from trees to graphs of treewidth $>1$); the collection of these methods is known as \emph{Non-serial dynamic programming} \cite{1972nonserial}. 

In this model, a tree-sequential problem can be specified by a finite  (bottom-up) tree automaton $A=(Q,\Sigma,\delta, q_0, Q_F)$, where $Q$ is a finite set of states, $\Sigma$ is a ranked alphabet, $\delta$ is a set of transition rules of the form $(q_1,\dots, q_n,a,q)$ where $q_1,\dots,q_n,q\in Q$ and $a\in\Sigma$, $q_0\in Q$ is the initial state, $Q_F\subseteq Q$ is a set of final states. The problem specification also includes a cost function. The set $L(A)$ of trees accepted by the tree automaton $A$ defines the set of feasible solutions. The minimization problem seeks to minimize the cost function over the set $L(A)$ of feasible solutions.

We will consider the simple case of a positive additive cost function that always equals zero in the initial state. An additive cost function can be defined via an incremental cost function $I\colon Q^*\times \Sigma\to \mathbb{R}$,  where $Q^*$ consists of tuples of states in $Q$ of the form $(q_1,\dots,q_n)$. $I(q_1,\dots,q_n,a)$ can be viewed as the cost of attaching $n$ child subtrees to a new root of symbol $a$. While it might seem restrictive to require an additive structure on the cost function, this simple case does cover many important problems admitting a DP-algorithm, for instance, Travelling Salesman \cite{BellmanTSP, HeldTSP}, Knapsack \cite{Kellerer2004}, or Levenshtein distance \cite{WagnerFischer}. 

In this case, the answer of the minimization problem can be shown to be equal to $\min_{q\in Q_F} D(q)$, where $D:Q\to\mathbb{R}_{\geq 0}$ is defined by the following recurrence equations:
\begin{equation}\label{equation:DP}
\begin{aligned}
&D(q_0)=0\,,\\
&\text{for}\;q\neq q_0\,,\quad D(q)=\min_{(q_1,\dots,q_n,a,q)\in \delta}\;\sum_{1\leq i\leq n} D(q_i)+I(q_1,\dots,q_n,a)\,.\end{aligned}
\end{equation}
A \emph{dynamic programming algorithm} for the minimization problem corresponds to an algorithm that computes $D$; the function $D$ is commonly called a \emph{dynamic programming table} (a DP-tabled, also called a DP-array, or a DP-matrix). Needless to say, such an algorithm does not exist in general for given arbitrary tree automata and cost functions \cite{IBARAKI197484}. 

Using an algebraic approach, Gnesi and Montanari \cite{Gnesi} have shown that solving the functional \cref{equation:DP} corresponds to finding the solution subtrees of a general AND/OR graph. An important special case in which DP-algorithms exist is when the underlying AND/OR graph is acyclic.

When a fixed tree is given as an input to the problem, the underlying AND/OR graph is acyclic and decomposable (that is, it is an ad-AND/OR graph). Such problems are hence naturally solvable by DP-algorithms. These algorithms are known in folklore under the name \emph{Dynamic programming on a tree}. Many graph-theoretical problems (e.g., maximum matching, longest path) can be solved optimally on trees by DP-algorithms. Numerous real-world applications also rely on DP-algorithms on trees; examples can be found, for instance, in Data Science \cite{rokach2008data}, Computer Vision \cite{Felzenszwalb,Veksler}, and Computational Biology \cite{Bansal12,Donati2015}.

\paragraph*{Explicit construction of the ad-AND/OR graph for DP on a fixed tree}
Due to its usefulness for the examples that we will develop next, in the case of DP on a fixed tree, an explicit construction of the ad-AND/OR graph from \cref{equation:DP} is described below.  The construction is done in two steps. In the first step, we build a graph in which every node retains an additional attribute, its  \emph{value}, and every \gor\ node is labeled by a state $q\in Q$. In the second step, we \emph{prune} the graph by removing nodes that do not yield optimal values.
\begin{enumerate}
\item For each $(q_0,a,q)\in \delta$, create a goal node of value $0$ labeled by $q$. Then, for each $q\neq q_0$ in post-order,
\begin{enumerate}[i.]
\item For each $(q_1,\dots,q_n,a,q)\in \delta$, create an AND node, connect it to the $n$ \gor\ nodes labeled by $q_1,\dots,q_n$. Its value is equal to the sum of the values of its children, plus $I(q_1,\dots,q_n,a)$.
\item Create a single OR node, connect it to every AND node created in the previous step. 
Its label is $q$, and its value is the minimum of the values of its children.
\end{enumerate}
\item For each $q\in Q_F$, remove the OR node labeled by $q$ unless its value is equal to $\min_{q\in Q_F} G(q)$. For each OR node $s$, remove the arc to its AND-child node $s_i$ if the value of $s_i$ is not equal to the value of $s$. Finally, remove recursively all AND nodes without incoming arcs.
\end{enumerate}

\subsection{Examples}
\subsubsection{Optimal tree coloring problem} 

\paragraph*{Description}
A prototypical problem that fits into the framework of DP on a tree is \textsc{Optimal tree coloring}, that is, finding an optimal node-coloring of the input tree. Many problems of practical interest reduce to \textsc{Optimal tree coloring}; three concrete examples are given later in this section.

If $T$ is the input (rooted, ordered) tree and $C$ is the set of colors, such a problem seeks a coloring $\phi\colon V(T)\to C$ that minimizes the cost function. There can be many constraints on the coloring function: some nodes of $T$ may be forced to have a certain color, the possible colors of a node may depend on the colors of its descendants. In our tree-sequential dynamic programming formalism, a tree automaton and a cost function are given as part of the input. The tree automaton defines the set $L(A)$ of feasible coloring functions satisfying all those constraints. A state $q$ can be interpreted as a colored subtree of $T$ with a particular root color; the unique initial state is an empty coloring and transitions into a colored leaf of $T$; a final state corresponds to a fully colored $T$ with a particular root color. A commonly used form of cost functions considers the (possibly weighted) sum over the edges of the tree of the cost of putting two colors on each end of an edge, that is, an incremental cost function $I$ of the form $I(q_1,\dots,q_n,a,p)=\sum_{1\leq i\leq n} p(a_i,a)$ where $a_i$ is the color of the root of the subtree in state $q_i$ and $p\colon C^2\to\mathbb{R}_{\geq 0}$ is a function that gives the cost of putting two colors at each end of an edge.

\paragraph*{Equivalence relations on the set of solutions}
A possible strategy to define equivalence classes on the solution space of the \textsc{Optimal tree coloring} problem is to consider some colors to be locally equivalent on a node. In practical applications (see the next section), the space of colors can be quite large. Even though the precise colors of each node are necessary for correctly computing the cost function, when the solutions are analyzed by a human expert, it can be desirable to omit the colors and just look at whether the color of a node belongs to some group of colors. Therefore, this kind of equivalence relations is natural in many situations. Our \cref{definition:equivalence-classes} of equivalence classes of an e-colored ad-AND/OR graph deals exactly with equivalence relations of this type.

Let $e$ be a function that maps a color $c$ to its ``color group'' $e(c)$. Two solutions of the \textsc{Optimal tree coloring} problem $\phi_1,\phi_2\colon V(T)\to C$ are said to be \emph{equivalent} if $\forall u\in V(T)$, $e(\phi_1(u))=e(\phi_2(u))$. Let $G$ be the ad-AND/OR graph associated with this instance. For each \gor\ node $s$ of $G$ labeled with the state $q$, where $q$ is interpreted as a colored subtree of $T$ with a particular root color $c$, color the node $s$ with $e(c)$. 
Then $G$ is e-colored and $\mathbb{C}(G)$ corresponds to the set of equivalence classes of the solutions of the instance. Notice that the constraint we had on the e-coloring of an ad-AND/OR graph is naturally satisfied by any meaningful function $e$ because in a DP setting we only consider ordered trees: the $i$-th and $j$-th children of a node of $T$ cannot be in the same color group unless $i=j$.

\paragraph*{Concrete examples of tree coloring problems}

\subparagraph*{Example 1} 
The first example is related to the alignment of gene sequences on a phylogenetic tree \cite{Sankoff75}. The input is a tree $T$, a set $\Sigma$ of letters (DNA alphabet or protein alphabet), a function that labels each leaf node of $T$ with a letter, and a distance function $d\colon\Sigma^2\to \mathbb{R}_{\geq 0}$ between two letters. The goal is to extend the leaf labeling to a full labeling $\phi\colon V(T)\to \Sigma$ such that the sum of the distances over the edges of $T$ is minimized. Defining equivalence relations of the solutions based on a grouping of the letters is uncontrived in this problem: for the DNA or protein alphabet, the letters can be subdivided into structurally similar groups.

\subparagraph*{Example 2} 
The \textsc{Frequency assignment} problems are a family of problems that naturally arise in telecommunication networks, and that have been extensively studied in graph theory as a generalization of graph coloring known as the T-coloring problem \cite{Hale,Roberts,Tesman}. In the variant called the list T-coloring, the input is a graph $G$ representing the interference between radio stations, a set $C$ of colors, a function $S$ that gives for each vertex $v\in V(G)$ a set $S(v)\subseteq C$ of colors (possible frequencies for a station), and a set $T\subseteq C^2$ of forbidden pairs of colors (interfering frequencies). The goal is to find a coloring $\phi\colon V(G)\to C$ such that a $\forall v\in V(G)$, $\phi(v)\in S(v)$, and $\forall (u,v)\in E(G)$, $(\phi(u),\phi(v))\not\in T$. While this problem is hard in general, it can be solved by DP when the underlying graph is a tree. In this case, we can enumerate colorings of the input trees without any forbidden pair of colors on the edges. Defining equivalence relations by grouping some of the colors together (similar frequencies) can be a practical way of reducing the size of the output. 

\subparagraph*{Example 3} 
The \textsc{Tree Reconciliation} problem is the main  method for analyzing  the co-evolution of two sets of species, the hosts and their parasites \cite{Page:TangledTrees2003}. The input are two phylogenetic trees $H, P$ (of the hosts and of the parasites, respectively), together with a mapping $\phi_0\colon  \textsf{Leaves}(P) \to \textsf{Leaves}(H)$ that reflects the present-day parasite infections. What needs to be enumerated are then all past associations, that is, all mappings of the non-leaf nodes of the parasite tree to the nodes of the host tree that optimize a function which overall represents the sum of the number of different possible ``events'' weighted by the inverse of their estimated probability. The number of optimal solutions is often huge and, by applying our \cref{alg:Newone}, the enumeration of biologically inspired equivalence classes have allowed a significant reduction (in some cases from $10^{42}$ to only $96$ classes) of the size of the output while still preserving the important biological information (see \cite{capybara}).

\subsubsection{Dynamic programming on tree decomposition of a graph} 

Many graph problems can be solved in  polynomial time with a dynamic programming algorithm when the input graph has bounded treewidth (see for example \cite{Bodlaender}). The underlying idea is that, given a tree decomposition of a graph, the  dynamic programming algorithm  traverses the nodes (bags) of the decomposition and consecutively solves the respective sub-problems.
For vertex subset optimization problems, given a bag $X$, a dynamic programming algorithm generally computes for each $Z\subseteq X$ the optimal solution of the sub-problem whose intersection with $X$ is $Z$. In this context, we could define two solutions to be equivalent if they intersect each bag of the decomposition in an ``equivalent'' way. The equivalence relation on the solutions is then defined by an equivalence relation over the subsets of each bag, and two solutions $S_1$ and $S_2$ are equivalent if for all bags $X$ of the decomposition, $S_1\cap X$ is equivalent to $S_2\cap X$. 

One of the simplest examples is to consider that  all the nonempty subsets of vertices of a bag are equivalent. Thus, what we are interested in is whether a solution ``hits'' a bag ({i.e.}, whether it has a nonempty intersection with the vertices in the bag). Consequently, two solutions would be considered equivalent if they hit the same bags.

We can also consider two subsets of a bag to be equivalent if they have the same size. In this case, two solutions would be equivalent if each bag contains the same number of vertices in the two solutions. 

We believe that considering solutions in the way they are distributed along the tree decomposition of a graph could give a good overview of the diversity of the solution space.

\section{Conclusion and perspectives}
\label{sec:conclusions}

In this paper, we provide a  general framework 
for the enumeration of equivalence classes of solutions in polynomial delay for a wide variety of contexts. This work opens a door to different research directions. 
 
  It would be interesting to ask whether we can efficiently enumerate groups of solutions that result from classical clustering procedures, or one representative per group.
 Moreover, in this paper we heavily rely on the decomposability property of the structure of the solution space. It remains open whether the problem of enumerating equivalence classes is hard without this restriction.

\bibliographystyle{plainurl}
\bibliography{references}

\end{document}

%% file: pic1.pdf_tex
\begingroup%
  \makeatletter%
  \providecommand\color[2][]{%
    \errmessage{(Inkscape) Color is used for the text in Inkscape, but the package 'color.sty' is not loaded}%
    \renewcommand\color[2][]{}%
  }%
  \providecommand\transparent[1]{%
    \errmessage{(Inkscape) Transparency is used (non-zero) for the text in Inkscape, but the package 'transparent.sty' is not loaded}%
    \renewcommand\transparent[1]{}%
  }%
  \providecommand\rotatebox[2]{#2}%
  \newcommand*\fsize{\dimexpr\f@size pt\relax}%
  \newcommand*\lineheight[1]{\fontsize{\fsize}{#1\fsize}\selectfont}%
  \ifx\svgwidth\undefined%
    \setlength{\unitlength}{496.86034716bp}%
    \ifx\svgscale\undefined%
      \relax%
    \else%
      \setlength{\unitlength}{\unitlength * \real{\svgscale}}%
    \fi%
  \else%
    \setlength{\unitlength}{\svgwidth}%
  \fi%
  \global\let\svgwidth\undefined%
  \global\let\svgscale\undefined%
  \makeatother%
  \begin{picture}(1,0.52712335)%
    \lineheight{1}%
    \setlength\tabcolsep{0pt}%
    \put(0,0){\includegraphics[width=\unitlength,page=1]{pic1.pdf}}%
  \end{picture}%
\endgroup%

%% file: pic2.pdf_tex
\begingroup%
  \makeatletter%
  \providecommand\color[2][]{%
    \errmessage{(Inkscape) Color is used for the text in Inkscape, but the package 'color.sty' is not loaded}%
    \renewcommand\color[2][]{}%
  }%
  \providecommand\transparent[1]{%
    \errmessage{(Inkscape) Transparency is used (non-zero) for the text in Inkscape, but the package 'transparent.sty' is not loaded}%
    \renewcommand\transparent[1]{}%
  }%
  \providecommand\rotatebox[2]{#2}%
  \newcommand*\fsize{\dimexpr\f@size pt\relax}%
  \newcommand*\lineheight[1]{\fontsize{\fsize}{#1\fsize}\selectfont}%
  \ifx\svgwidth\undefined%
    \setlength{\unitlength}{541.9532181bp}%
    \ifx\svgscale\undefined%
      \relax%
    \else%
      \setlength{\unitlength}{\unitlength * \real{\svgscale}}%
    \fi%
  \else%
    \setlength{\unitlength}{\svgwidth}%
  \fi%
  \global\let\svgwidth\undefined%
  \global\let\svgscale\undefined%
  \makeatother%
  \begin{picture}(1,0.47972619)%
    \lineheight{1}%
    \setlength\tabcolsep{-1pt}%
    \put(0,0.003){\includegraphics[width=\unitlength,page=1]{pic2.pdf}}%
    \put(0.1838735,0.4506975){\color[rgb]{0,0,0}\makebox(0,0)[lt]{\lineheight{1.25}\smash{\begin{tabular}[t]{l}$a$\end{tabular}}}}%
    \put(0.02215517,0.01831252){\color[rgb]{0,0,0}\makebox(0,0)[lt]{\lineheight{1.25}\smash{\begin{tabular}[t]{l}$c$\end{tabular}}}}%
    \put(0.09470269,0.01831252){\color[rgb]{0,0,0}\makebox(0,0)[lt]{\lineheight{1.25}\smash{\begin{tabular}[t]{l}$b$\end{tabular}}}}%
    \put(0.20352391,0.01831252){\color[rgb]{0,0,0}\makebox(0,0)[lt]{\lineheight{1.25}\smash{\begin{tabular}[t]{l}$a$\end{tabular}}}}%
    \put(0.31283635,0.01831252){\color[rgb]{0,0,0}\makebox(0,0)[lt]{\lineheight{1.25}\smash{\begin{tabular}[t]{l}$b$\end{tabular}}}}%
    \put(0.38395172,0.01831252){\color[rgb]{0,0,0}\makebox(0,0)[lt]{\lineheight{1.25}\smash{\begin{tabular}[t]{l}$c$\end{tabular}}}}%
    \put(0.02370108,0.22487194){\color[rgb]{0,0,0}\makebox(0,0)[lt]{\lineheight{1.25}\smash{\begin{tabular}[t]{l}$y$\end{tabular}}}}%
    \put(0.10382556,0.22487194){\color[rgb]{0,0,0}\makebox(0,0)[lt]{\lineheight{1.25}\smash{\begin{tabular}[t]{l}$x$\end{tabular}}}}%
    \put(0.18395008,0.22487194){\color[rgb]{0,0,0}\makebox(0,0)[lt]{\lineheight{1.25}\smash{\begin{tabular}[t]{l}$y$\end{tabular}}}}%
    \put(0.26374126,0.22487194){\color[rgb]{0,0,0}\makebox(0,0)[lt]{\lineheight{1.25}\smash{\begin{tabular}[t]{l}$x$\end{tabular}}}}%
    \put(0.34419914,0.22487194){\color[rgb]{0,0,0}\makebox(0,0)[lt]{\lineheight{1.25}\smash{\begin{tabular}[t]{l}$y$\end{tabular}}}}%
    \put(0,0.003){\includegraphics[width=\unitlength,page=2]{pic2.pdf}}%
    \put(0.62659708,0.42309585){\color[rgb]{0,0,0}\makebox(0,0)[lt]{\lineheight{1.25}\smash{\begin{tabular}[t]{l}$a$\end{tabular}}}}%
    \put(0,0.003){\includegraphics[width=\unitlength,page=3]{pic2.pdf}}%
    \put(0.5875518,0.35291319){\color[rgb]{0,0,0}\makebox(0,0)[lt]{\lineheight{1.25}\smash{\begin{tabular}[t]{l}$x$\end{tabular}}}}%
    \put(0,0.003){\includegraphics[width=\unitlength,page=4]{pic2.pdf}}%
    \put(0.66564236,0.35291319){\color[rgb]{0,0,0}\makebox(0,0)[lt]{\lineheight{1.25}\smash{\begin{tabular}[t]{l}$y$\end{tabular}}}}%
    \put(0,0.003){\includegraphics[width=\unitlength,page=5]{pic2.pdf}}%
    \put(0.55692032,0.28299301){\color[rgb]{0,0,0}\makebox(0,0)[lt]{\lineheight{1.25}\smash{\begin{tabular}[t]{l}$a$\end{tabular}}}}%
    \put(0,0.003){\includegraphics[width=\unitlength,page=6]{pic2.pdf}}%
    \put(0.69627385,0.28299301){\color[rgb]{0,0,0}\makebox(0,0)[lt]{\lineheight{1.25}\smash{\begin{tabular}[t]{l}$c$\end{tabular}}}}%
    \put(0,0.003){\includegraphics[width=\unitlength,page=7]{pic2.pdf}}%
    \put(0.62659707,0.28299301){\color[rgb]{0,0,0}\makebox(0,0)[lt]{\lineheight{1.25}\smash{\begin{tabular}[t]{l}$b$\end{tabular}}}}%
    \put(0,0.003){\includegraphics[width=\unitlength,page=8]{pic2.pdf}}%
    \put(0.84085556,0.42309585){\color[rgb]{0,0,0}\makebox(0,0)[lt]{\lineheight{1.25}\smash{\begin{tabular}[t]{l}$a$\end{tabular}}}}%
    \put(0,0.003){\includegraphics[width=\unitlength,page=9]{pic2.pdf}}%
    \put(0.80181031,0.35291319){\color[rgb]{0,0,0}\makebox(0,0)[lt]{\lineheight{1.25}\smash{\begin{tabular}[t]{l}$x$\end{tabular}}}}%
    \put(0,0.003){\includegraphics[width=\unitlength,page=10]{pic2.pdf}}%
    \put(0.87990087,0.35291319){\color[rgb]{0,0,0}\makebox(0,0)[lt]{\lineheight{1.25}\smash{\begin{tabular}[t]{l}$y$\end{tabular}}}}%
    \put(0,0.003){\includegraphics[width=\unitlength,page=11]{pic2.pdf}}%
    \put(0.7711788,0.28299301){\color[rgb]{0,0,0}\makebox(0,0)[lt]{\lineheight{1.25}\smash{\begin{tabular}[t]{l}$c$\end{tabular}}}}%
    \put(2.27949343,-0.29322631){\color[rgb]{0,0,0}\makebox(0,0)[lt]{\begin{minipage}{0.84077183\unitlength}\raggedright \end{minipage}}}%
    \put(0,0.003){\includegraphics[width=\unitlength,page=12]{pic2.pdf}}%
    \put(0.91053233,0.28299301){\color[rgb]{0,0,0}\makebox(0,0)[lt]{\lineheight{1.25}\smash{\begin{tabular}[t]{l}$c$\end{tabular}}}}%
    \put(0,0.003){\includegraphics[width=\unitlength,page=13]{pic2.pdf}}%
    \put(0.84085555,0.28299301){\color[rgb]{0,0,0}\makebox(0,0)[lt]{\lineheight{1.25}\smash{\begin{tabular}[t]{l}$b$\end{tabular}}}}%
    \put(0,0.003){\includegraphics[width=\unitlength,page=14]{pic2.pdf}}%
    \put(0.56235382,0.19687347){\color[rgb]{0,0,0}\makebox(0,0)[lt]{\lineheight{1.25}\smash{\begin{tabular}[t]{l}$a$\end{tabular}}}}%
    \put(0,0.003){\includegraphics[width=\unitlength,page=15]{pic2.pdf}}%
    \put(0.52330854,0.12669081){\color[rgb]{0,0,0}\makebox(0,0)[lt]{\lineheight{1.25}\smash{\begin{tabular}[t]{l}$x$\end{tabular}}}}%
    \put(0,0.003){\includegraphics[width=\unitlength,page=16]{pic2.pdf}}%
    \put(0.6013991,0.12669081){\color[rgb]{0,0,0}\makebox(0,0)[lt]{\lineheight{1.25}\smash{\begin{tabular}[t]{l}$y$\end{tabular}}}}%
    \put(0,0.003){\includegraphics[width=\unitlength,page=17]{pic2.pdf}}%
    \put(0.49267706,0.05677063){\color[rgb]{0,0,0}\makebox(0,0)[lt]{\lineheight{1.25}\smash{\begin{tabular}[t]{l}$a$\end{tabular}}}}%
    \put(0,0.003){\includegraphics[width=\unitlength,page=18]{pic2.pdf}}%
    \put(0.56235381,0.05677063){\color[rgb]{0,0,0}\makebox(0,0)[lt]{\lineheight{1.25}\smash{\begin{tabular}[t]{l}$b$\end{tabular}}}}%
    \put(0,0.003){\includegraphics[width=\unitlength,page=19]{pic2.pdf}}%
    \put(0.90509882,0.19687347){\color[rgb]{0,0,0}\makebox(0,0)[lt]{\lineheight{1.25}\smash{\begin{tabular}[t]{l}$a$\end{tabular}}}}%
    \put(0,0.003){\includegraphics[width=\unitlength,page=20]{pic2.pdf}}%
    \put(0.86605354,0.12669081){\color[rgb]{0,0,0}\makebox(0,0)[lt]{\lineheight{1.25}\smash{\begin{tabular}[t]{l}$x$\end{tabular}}}}%
    \put(0,0.003){\includegraphics[width=\unitlength,page=21]{pic2.pdf}}%
    \put(0.9441441,0.12669081){\color[rgb]{0,0,0}\makebox(0,0)[lt]{\lineheight{1.25}\smash{\begin{tabular}[t]{l}$y$\end{tabular}}}}%
    \put(0,0.003){\includegraphics[width=\unitlength,page=22]{pic2.pdf}}%
    \put(0.83542205,0.05677063){\color[rgb]{0,0,0}\makebox(0,0)[lt]{\lineheight{1.25}\smash{\begin{tabular}[t]{l}$a$\end{tabular}}}}%
    \put(0,0.003){\includegraphics[width=\unitlength,page=23]{pic2.pdf}}%
    \put(0.97477559,0.05677063){\color[rgb]{0,0,0}\makebox(0,0)[lt]{\lineheight{1.25}\smash{\begin{tabular}[t]{l}$c$\end{tabular}}}}%
    \put(0,0.003){\includegraphics[width=\unitlength,page=24]{pic2.pdf}}%
    \put(0.90509881,0.05677063){\color[rgb]{0,0,0}\makebox(0,0)[lt]{\lineheight{1.25}\smash{\begin{tabular}[t]{l}$b$\end{tabular}}}}%
    \put(0,0.003){\includegraphics[width=\unitlength,page=25]{pic2.pdf}}%
    \put(0.7331183,0.19687347){\color[rgb]{0,0,0}\makebox(0,0)[lt]{\lineheight{1.25}\smash{\begin{tabular}[t]{l}$a$\end{tabular}}}}%
    \put(0,0.003){\includegraphics[width=\unitlength,page=26]{pic2.pdf}}%
    \put(0.69407302,0.12669081){\color[rgb]{0,0,0}\makebox(0,0)[lt]{\lineheight{1.25}\smash{\begin{tabular}[t]{l}$x$\end{tabular}}}}%
    \put(0,0.003){\includegraphics[width=\unitlength,page=27]{pic2.pdf}}%
    \put(0.77216358,0.12669081){\color[rgb]{0,0,0}\makebox(0,0)[lt]{\lineheight{1.25}\smash{\begin{tabular}[t]{l}$y$\end{tabular}}}}%
    \put(0,0.003){\includegraphics[width=\unitlength,page=28]{pic2.pdf}}%
    \put(0.66344153,0.05677063){\color[rgb]{0,0,0}\makebox(0,0)[lt]{\lineheight{1.25}\smash{\begin{tabular}[t]{l}$c$\end{tabular}}}}%
    \put(0,0.003){\includegraphics[width=\unitlength,page=29]{pic2.pdf}}%
    \put(0.73311829,0.05677063){\color[rgb]{0,0,0}\makebox(0,0)[lt]{\lineheight{1.25}\smash{\begin{tabular}[t]{l}$b$\end{tabular}}}}%
    \put(0,0.003){\includegraphics[width=\unitlength,page=30]{pic2.pdf}}%
  \end{picture}%
\endgroup%

%% file: pic3.pdf_tex
\begingroup%
  \makeatletter%
  \providecommand\color[2][]{%
    \errmessage{(Inkscape) Color is used for the text in Inkscape, but the package 'color.sty' is not loaded}%
    \renewcommand\color[2][]{}%
  }%
  \providecommand\transparent[1]{%
    \errmessage{(Inkscape) Transparency is used (non-zero) for the text in Inkscape, but the package 'transparent.sty' is not loaded}%
    \renewcommand\transparent[1]{}%
  }%
  \providecommand\rotatebox[2]{#2}%
  \newcommand*\fsize{\dimexpr\f@size pt\relax}%
  \newcommand*\lineheight[1]{\fontsize{\fsize}{#1\fsize}\selectfont}%
  \ifx\svgwidth\undefined%
    \setlength{\unitlength}{545.00589951bp}%
    \ifx\svgscale\undefined%
      \relax%
    \else%
      \setlength{\unitlength}{\unitlength * \real{\svgscale}}%
    \fi%
  \else%
    \setlength{\unitlength}{\svgwidth}%
  \fi%
  \global\let\svgwidth\undefined%
  \global\let\svgscale\undefined%
  \makeatother%
  \begin{picture}(1,0.47679732)%
    \lineheight{1}%
    \setlength\tabcolsep{0pt}%
    \put(0,0){\includegraphics[width=\unitlength,page=1]{pic3.pdf}}%
    \put(0.73804705,0.41154635){\color[rgb]{0,0,0}\makebox(0,0)[lt]{\lineheight{1.25}\smash{\begin{tabular}[t]{l}$a$\end{tabular}}}}%
    \put(0,0){\includegraphics[width=\unitlength,page=2]{pic3.pdf}}%
    \put(0.69922046,0.34175679){\color[rgb]{0,0,0}\makebox(0,0)[lt]{\lineheight{1.25}\smash{\begin{tabular}[t]{l}$x$\end{tabular}}}}%
    \put(0,0){\includegraphics[width=\unitlength,page=3]{pic3.pdf}}%
    \put(0.77687363,0.34175679){\color[rgb]{0,0,0}\makebox(0,0)[lt]{\lineheight{1.25}\smash{\begin{tabular}[t]{l}$y$\end{tabular}}}}%
    \put(0,0){\includegraphics[width=\unitlength,page=4]{pic3.pdf}}%
    \put(0.69922046,0.27222827){\color[rgb]{0,0,0}\makebox(0,0)[lt]{\lineheight{1.25}\smash{\begin{tabular}[t]{l}$a$\end{tabular}}}}%
    \put(0,0){\includegraphics[width=\unitlength,page=5]{pic3.pdf}}%
    \put(0.77687363,0.27222824){\color[rgb]{0,0,0}\makebox(0,0)[lt]{\lineheight{1.25}\smash{\begin{tabular}[t]{l}$c$\end{tabular}}}}%
    \put(0,0){\includegraphics[width=\unitlength,page=6]{pic3.pdf}}%
    \put(0.93510735,0.41154635){\color[rgb]{0,0,0}\makebox(0,0)[lt]{\lineheight{1.25}\smash{\begin{tabular}[t]{l}$a$\end{tabular}}}}%
    \put(0,0){\includegraphics[width=\unitlength,page=7]{pic3.pdf}}%
    \put(0.89628077,0.34175679){\color[rgb]{0,0,0}\makebox(0,0)[lt]{\lineheight{1.25}\smash{\begin{tabular}[t]{l}$x$\end{tabular}}}}%
    \put(0,0){\includegraphics[width=\unitlength,page=8]{pic3.pdf}}%
    \put(0.97393393,0.34175679){\color[rgb]{0,0,0}\makebox(0,0)[lt]{\lineheight{1.25}\smash{\begin{tabular}[t]{l}$y$\end{tabular}}}}%
    \put(0,0){\includegraphics[width=\unitlength,page=9]{pic3.pdf}}%
    \put(0.89628074,0.26722827){\color[rgb]{0,0,0}\makebox(0,0)[lt]{\lineheight{1.25}\smash{\begin{tabular}[t]{l}$b$\end{tabular}}}}%
    \put(0,0){\includegraphics[width=\unitlength,page=10]{pic3.pdf}}%
    \put(0.97393393,0.26722825){\color[rgb]{0,0,0}\makebox(0,0)[lt]{\lineheight{1.25}\smash{\begin{tabular}[t]{l}$d$\end{tabular}}}}%
    \put(0,0){\includegraphics[width=\unitlength,page=11]{pic3.pdf}}%
    \put(0.73163642,0.15550472){\color[rgb]{0,0,0}\makebox(0,0)[lt]{\lineheight{1.25}\smash{\begin{tabular}[t]{l}$a$\end{tabular}}}}%
    \put(0,0){\includegraphics[width=\unitlength,page=12]{pic3.pdf}}%
    \put(0.69280984,0.08571517){\color[rgb]{0,0,0}\makebox(0,0)[lt]{\lineheight{1.25}\smash{\begin{tabular}[t]{l}$x$\end{tabular}}}}%
    \put(0,0){\includegraphics[width=\unitlength,page=13]{pic3.pdf}}%
    \put(0.770463,0.08571517){\color[rgb]{0,0,0}\makebox(0,0)[lt]{\lineheight{1.25}\smash{\begin{tabular}[t]{l}$y$\end{tabular}}}}%
    \put(0,0){\includegraphics[width=\unitlength,page=14]{pic3.pdf}}%
    \put(0.69280981,0.01618665){\color[rgb]{0,0,0}\makebox(0,0)[lt]{\lineheight{1.25}\smash{\begin{tabular}[t]{l}$a$\end{tabular}}}}%
    \put(0,0){\includegraphics[width=\unitlength,page=15]{pic3.pdf}}%
    \put(0.770463,0.01318662){\color[rgb]{0,0,0}\makebox(0,0)[lt]{\lineheight{1.25}\smash{\begin{tabular}[t]{l}$d$\end{tabular}}}}%
    \put(0,0){\includegraphics[width=\unitlength,page=16]{pic3.pdf}}%
    \put(0.93609029,0.15550472){\color[rgb]{0,0,0}\makebox(0,0)[lt]{\lineheight{1.25}\smash{\begin{tabular}[t]{l}$a$\end{tabular}}}}%
    \put(0,0){\includegraphics[width=\unitlength,page=17]{pic3.pdf}}%
    \put(0.89726371,0.08571517){\color[rgb]{0,0,0}\makebox(0,0)[lt]{\lineheight{1.25}\smash{\begin{tabular}[t]{l}$x$\end{tabular}}}}%
    \put(0,0){\includegraphics[width=\unitlength,page=18]{pic3.pdf}}%
    \put(0.97491687,0.08571517){\color[rgb]{0,0,0}\makebox(0,0)[lt]{\lineheight{1.25}\smash{\begin{tabular}[t]{l}$y$\end{tabular}}}}%
    \put(0,0){\includegraphics[width=\unitlength,page=19]{pic3.pdf}}%
    \put(0.89726368,0.01328665){\color[rgb]{0,0,0}\makebox(0,0)[lt]{\lineheight{1.25}\smash{\begin{tabular}[t]{l}$b$\end{tabular}}}}%
    \put(0,0){\includegraphics[width=\unitlength,page=20]{pic3.pdf}}%
    \put(0.97491687,0.01618662){\color[rgb]{0,0,0}\makebox(0,0)[lt]{\lineheight{1.25}\smash{\begin{tabular}[t]{l}$c$\end{tabular}}}}%
    \put(0,0){\includegraphics[width=\unitlength,page=21]{pic3.pdf}}%
    \put(0.64016749,0.46352875){\color[rgb]{0,0,0}\makebox(0,0)[lt]{\lineheight{1.25}\smash{\begin{tabular}[t]{l}Admissible combinations:\end{tabular}}}}%
    \put(0.6400421,0.20383302){\color[rgb]{0,0,0}\makebox(0,0)[lt]{\lineheight{1.25}\smash{\begin{tabular}[t]{l}Non admissible combinations:\end{tabular}}}}%
    \put(0.15182612,0.42712635){\color[rgb]{0,0,0}\makebox(0,0)[lt]{\lineheight{1.25}\smash{\begin{tabular}[t]{l}$a$\end{tabular}}}}%
    \put(0.41656774,0.42712635){\color[rgb]{0,0,0}\makebox(0,0)[lt]{\lineheight{1.25}\smash{\begin{tabular}[t]{l}$a$\end{tabular}}}}%
    \put(0.0358194,0.22839156){\color[rgb]{0,0,0}\makebox(0,0)[lt]{\lineheight{1.25}\smash{\begin{tabular}[t]{l}$z$\end{tabular}}}}%
    \put(0.13472953,0.23094042){\color[rgb]{0,0,0}\makebox(0,0)[lt]{\lineheight{1.25}\smash{\begin{tabular}[t]{l}$y$\end{tabular}}}}%
    \put(0.23453572,0.22827505){\color[rgb]{0,0,0}\makebox(0,0)[lt]{\lineheight{1.25}\smash{\begin{tabular}[t]{l}$x$\end{tabular}}}}%
    \put(0.33401038,0.22827505){\color[rgb]{0,0,0}\makebox(0,0)[lt]{\lineheight{1.25}\smash{\begin{tabular}[t]{l}$x$\end{tabular}}}}%
    \put(0.43315356,0.23094042){\color[rgb]{0,0,0}\makebox(0,0)[lt]{\lineheight{1.25}\smash{\begin{tabular}[t]{l}$y$\end{tabular}}}}%
    \put(0.5300391,0.2283781){\color[rgb]{0,0,0}\makebox(0,0)[lt]{\lineheight{1.25}\smash{\begin{tabular}[t]{l}$w$\end{tabular}}}}%
    \put(0.13547314,0.0500568){\color[rgb]{0,0,0}\makebox(0,0)[lt]{\lineheight{1.25}\smash{\begin{tabular}[t]{l}$c$\end{tabular}}}}%
    \put(0.23445958,0.05006577){\color[rgb]{0,0,0}\makebox(0,0)[lt]{\lineheight{1.25}\smash{\begin{tabular}[t]{l}$a$\end{tabular}}}}%
    \put(0.33408205,0.04772742){\color[rgb]{0,0,0}\makebox(0,0)[lt]{\lineheight{1.25}\smash{\begin{tabular}[t]{l}$b$\end{tabular}}}}%
    \put(0.43247266,0.04770052){\color[rgb]{0,0,0}\makebox(0,0)[lt]{\lineheight{1.25}\smash{\begin{tabular}[t]{l}$d$\end{tabular}}}}%
    \put(0.10826544,0.4249672){\color[rgb]{0,0,0}\makebox(0,0)[lt]{\lineheight{1.25}\smash{\begin{tabular}[t]{l}\small 1\end{tabular}}}}%
    \put(0.36874753,0.42504335){\color[rgb]{0,0,0}\makebox(0,0)[lt]{\lineheight{1.25}\smash{\begin{tabular}[t]{l}\small 2\end{tabular}}}}%
    \put(0.04299746,0.33252551){\color[rgb]{0,0,0}\makebox(0,0)[lt]{\lineheight{1.25}\smash{\begin{tabular}[t]{l}\small 3\end{tabular}}}}%
    \put(0.17687545,0.33239559){\color[rgb]{0,0,0}\makebox(0,0)[lt]{\lineheight{1.25}\smash{\begin{tabular}[t]{l}\small 4\end{tabular}}}}%
    \put(0.31055683,0.33262856){\color[rgb]{0,0,0}\makebox(0,0)[lt]{\lineheight{1.25}\smash{\begin{tabular}[t]{l}\small 5\end{tabular}}}}%
    \put(0.44187911,0.33250761){\color[rgb]{0,0,0}\makebox(0,0)[lt]{\lineheight{1.25}\smash{\begin{tabular}[t]{l}\small 6\end{tabular}}}}%
    \put(-0.00584216,0.22629955){\color[rgb]{0,0,0}\makebox(0,0)[lt]{\lineheight{1.25}\smash{\begin{tabular}[t]{l}\small 7\end{tabular}}}}%
    \put(0.09304275,0.22641157){\color[rgb]{0,0,0}\makebox(0,0)[lt]{\lineheight{1.25}\smash{\begin{tabular}[t]{l}\small 8\end{tabular}}}}%
    \put(0.19035492,0.22642054){\color[rgb]{0,0,0}\makebox(0,0)[lt]{\lineheight{1.25}\smash{\begin{tabular}[t]{l}\small 9\end{tabular}}}}%
    \put(0.27714554,0.2263578){\color[rgb]{0,0,0}\makebox(0,0)[lt]{\lineheight{1.25}\smash{\begin{tabular}[t]{l}\small 10\end{tabular}}}}%
    \put(0.37858609,0.22621443){\color[rgb]{0,0,0}\makebox(0,0)[lt]{\lineheight{1.25}\smash{\begin{tabular}[t]{l}\small 11\end{tabular}}}}%
    \put(0.47806077,0.2262234){\color[rgb]{0,0,0}\makebox(0,0)[lt]{\lineheight{1.25}\smash{\begin{tabular}[t]{l}\small 12\end{tabular}}}}%
    \put(0.07564051,0.13865614){\color[rgb]{0,0,0}\makebox(0,0)[lt]{\lineheight{1.25}\smash{\begin{tabular}[t]{l}\small 13\end{tabular}}}}%
    \put(0.17629471,0.13853523){\color[rgb]{0,0,0}\makebox(0,0)[lt]{\lineheight{1.25}\smash{\begin{tabular}[t]{l}\small 14\end{tabular}}}}%
    \put(0.27753872,0.13866511){\color[rgb]{0,0,0}\makebox(0,0)[lt]{\lineheight{1.25}\smash{\begin{tabular}[t]{l}\small 15\end{tabular}}}}%
    \put(0.37524408,0.13866511){\color[rgb]{0,0,0}\makebox(0,0)[lt]{\lineheight{1.25}\smash{\begin{tabular}[t]{l}\small 16\end{tabular}}}}%
    \put(0.12703043,0.00519558){\color[rgb]{0,0,0}\makebox(0,0)[lt]{\lineheight{1.25}\smash{\begin{tabular}[t]{l}\small 17\end{tabular}}}}%
    \put(0.22680079,0.00531653){\color[rgb]{0,0,0}\makebox(0,0)[lt]{\lineheight{1.25}\smash{\begin{tabular}[t]{l}\small 18\end{tabular}}}}%
    \put(0.32630233,0.00533891){\color[rgb]{0,0,0}\makebox(0,0)[lt]{\lineheight{1.25}\smash{\begin{tabular}[t]{l}\small 19\end{tabular}}}}%
    \put(0.42614879,0.0053882){\color[rgb]{0,0,0}\makebox(0,0)[lt]{\lineheight{1.25}\smash{\begin{tabular}[t]{l}\small 20\end{tabular}}}}%
  \end{picture}%
\endgroup%